%% file: root.tex
\documentclass[11pt,a4paper]{article}
\usepackage[T1]{fontenc}
\usepackage[utf8]{inputenc}
\usepackage{authblk}
\usepackage{comment}
\title{Structural Multi-type Sequent Calculus for Inquisitive Logic}

\author[1]{Sabine Frittella}
\author[1]{Giuseppe Greco}
\author[1,2]{Alessandra Palmigiano}
\author[1]{Fan Yang\footnote{This research has been made possible by the NWO Vidi grant 016.138.314, by the NWO Aspasia grant 015.008.054, and by a Delft Technology Fellowship awarded in 2013. }}
\affil[1]{Delft University of Technology, Delft, The Netherlands}
\affil[2]{Department of Pure and Applied Mathematics, University of Johannesburg, South Africa}

\date{}

\input{packages}

\input{macros}

\begin{document}

\maketitle

\begin{abstract}
In this paper, we define a multi-type calculus for inquisitive logic, which is sound, complete and enjoys Belnap-style cut-elimination and subformula property.
Inquisitive logic is the logic of inquisitive semantics, a semantic framework developed by Groenendijk, Roelofsen and Ciardelli which captures both assertions and questions in natural language. Inquisitive logic is sound and complete w.r.t. the so-called state semantics (also known as team semantics). The Hilbert-style presentation of inquisitive logic is not closed under uniform substitution; indeed, some occurrences of formulas are restricted to a certain subclass of formulas, called flat formulas. This and other features make the quest for analytic calculi for this logic not straightforward. We develop a certain algebraic and order-theoretic analysis of the team semantics, which provides the guidelines for the design of a multi-type environment which accounts for two domains of interpretation, for flat and for general formulas, as well as for their interaction. This multi-type environment in its turn provides the semantic environment for the multi-type calculus for inquisitive logic we introduce in this paper.
\end{abstract}

\section{Introduction}

\label{sec:intro}
\input{intro}

\section{Inquisitive logic}
\label{sec:DL}
\input{syntax_semantics_flatness}

\section{Order-theoretic analysis and multi-type inquisitive logic}
\label{semanticanalysis}
\input{semanticanalysis}

\section{Structural sequent calculus for multi-type inquisitive logic}
\label{sec:formal}
\input{DisplayCalculusPorPropositionalDependenceLogic}

\section{Properties of the calculus}
\label{sec:properties}
\input{Soundness}

\section{Cut elimination}
\label{sec:cutelim}
\input{CutElimination}

\section{Conclusion}
\input{conclusions}

\newpage

\bibliographystyle{plain}
\bibliography{fan}

\appendix



\section{Completeness}
\label{sec:completeness}
\input{Completeness-Proof}

\end{document}

%% file: packages.tex
\usepackage{lscape}
\usepackage{pdflscape}

\usepackage{stmaryrd}

\usepackage{amscd}
\usepackage{amssymb,amsmath,amsthm}
\usepackage{mathptmx}
\usepackage{mathrsfs}
\usepackage{color}
\usepackage{xspace}
\usepackage{bussproofs}
\EnableBpAbbreviations

\usepackage{tikz}

  \usepackage{txfonts}

\usepackage{bpextra}

\usepackage{centernot}

\usepackage{hyperref}

\usepackage{cleveref}

\usepackage{relsize}

\usepackage{caption}

\usepackage{multirow}

\usepackage{array}
\newcolumntype{C}[1]{>{\centering\arraybackslash}p{#1}}

%% file: macros.tex
\newtheorem{theorem}{Theorem}[section]

\newtheorem{lemma}[theorem]{Lemma}

\newtheorem{definition}[theorem]{Definition}

\newcommand{\Inql}{\ensuremath{\mathbf{InqL}}\xspace}
\newcommand{\CPC}{\ensuremath{\mathbf{CPL}}\xspace}
\newcommand{\IPC}{\ensuremath{\mathbf{IPL}}\xspace}
\newcommand{\PD}{\ensuremath{\mathbf{PD}}\xspace}
\newcommand{\MPD}{\ensuremath{\mathbf{MPD}}\xspace}

\newcommand{\dep}{\ensuremath{=\!\!}\xspace}

\newcommand{\cand}{\ensuremath{\sqcap}\xspace}
\newcommand{\cor}{\ensuremath{\sqcup}\xspace}

\newcommand{\cra}{\ensuremath{\rightarrowtriangle}\xspace}
\newcommand{\cdra}{\ensuremath{\mapsto}\xspace}

\newcommand{\cneg}{\ensuremath{{\sim}}\xspace}

\newcommand{\CRA}{\ensuremath{\sqsupset}\xspace}

\newcommand{\iand}{\ensuremath{\wedge}\xspace}
\newcommand{\ior}{\ensuremath{\vee}\xspace}
\newcommand{\bigior}{\ensuremath{\bigvee}\xspace}
\newcommand{\ira}{\ensuremath{\to}\xspace}
\newcommand{\idra}{\ensuremath{\rightarrowtail}\xspace}

\newcommand{\ineg}{\ensuremath{\neg}\xspace}

\newcommand{\bh}{\ensuremath{{\downarrow}}\xspace}
\newcommand{\BH}{\ensuremath{{\Downarrow}}\xspace}
\newcommand{\hb}{\ensuremath{{\mathrm{f}}}\xspace}
\newcommand{\HB}{\ensuremath{\mathrm{F}}\xspace}

\newcommand{\marginnote}[1]{\marginpar{\raggedright\tiny{#1}}}

\def\fCenter{{\mbox{$\ \vdash\ $}}}

\newcommand\val[1]{{\lbrack\!\lbrack} {#1}{\rbrack\!\rbrack}}

\newcommand{\MP}{\ensuremath{\mathsf{MP}}\xspace}

\def\mc{\multicolumn}


%% file: intro.tex
 
Inquisitive logic is the logic of inquisitive semantics \cite{GR_09,InquiLog}, a semantic framework that captures both assertions and questions in natural language. 
In this framework, sentences express proposals to enhance the common ground of a conversation. The inquisitive content of a sentence is understood as an issue raised by an utterance of the sentence. A distinguishing feature of inquisitive logic is that formulas are evaluated on \emph{information states}, i.e., a set of possible worlds, instead of single possible worlds. Inquisitive logic defines a relation of \emph{support} between information states and sentences, where the idea is that in uttering a sentence $\phi$, a speaker proposes to enhance the current common ground to one that supports $\phi$.

Closely related to inquisitive logic  is \emph{dependence logic} \cite{Van07dl}, which is an extension of classical logic that characterizes the notion of ``dependence'' using the so-called \emph{team semantics} \cite{hodges1997a,hodges1997b}. The team semantics of dependence logic builds on the basis of the notion of \emph{team}, which, in the propositional logic context, is a \emph{set} of valuations.  Possible worlds can be identified with valuations. Therefore, an information state is essentially a team, and the state semantics that inquisitive logic adopts is essentially team semantics. Technically, it was observed in \cite{Yang_dissertation} that inquisitive logic  is essentially a variant of propositional dependence logic \cite{VY_PD} with the intuitionistic connectives introduced in \cite{AbramskyVaananen08}. It was further argued in \cite{Ciardelli2015} that the entailment relation of questions is a type of dependency relation considered in dependence logic.

Inquisitive logic was axiomatized in \cite{InquiLog}, and this axiomatization is not closed under uniform substitution, which is a hurdle for a smooth proof-theoretic treatment for inquisitive logic. In \cite{Sano_09}, a labelled calculus was introduced for an earlier version of inquisitive logic, defined on the basis of the so called pair semantics \cite{Groenendijk09, Mascarenhas_msc}. The calculus in \cite{Sano_09} makes use of extra linguistic labels which import the pair semantics for inquisitive logic  into the calculus. This calculus is sound, complete and cut free; however, the proof of the soundness of the rules is very involved, since the interpretation of the sequents is ad hoc, and only a semantic proof of cut elimination is given. 

Our contribution is a calculus designed on different principles than those of \cite{Sano_09}, and for the version of inquisitive logic based on state semantics. We tackle the hurdle of the non schematicity of the Hilbert-style presentation by designing the calculus for inquisitive logic in the style of a generalization of Belnap's display calculi, the so-called  {\em multi-type calculi}. These calculi have been introduced in \cite{Multitype, PDL}, as a proposal to support a proof-theoretic semantic account of Dynamic Logics \cite{GAV}. One important aspect of multi-type calculi is that various Belnap-style metatheorems have been given, which allow for a smooth syntactic proof of cut elimination. 

The multi-type environment we propose is motivated by an order-theoretic analysis of the team semantics for inquisitive logic, according to which, certain maps can be defined which make it possible for the different types to interact. The non schematicity of the axioms is accounted for by assigning different types to the restricted formulas and to the general formulas. Hence, closure under arbitrary substitution holds {\em within each type}. 

\paragraph{Structure of the paper.} In Section \ref{sec:DL}, needed preliminaries are collected on inquisitive logic. In Section \ref{semanticanalysis}, the order-theoretic analysis is given, which justifies the introduction of an expanded multi-type language, into which the original language of inquisitive logic can be embedded. In Section \ref{sec:formal}, the multi-type calculus for (the multi-type version of) inquisitive logic is introduced. In Section \ref{sec:properties}, two properties of the calculus are shown: soundness, and the fact that the calculus is powerful enough to capture the restricted type (i.e.\ the flat type) proof-theoretically. In Section \ref{sec:cutelim}, we give a syntactic proof of cut elimination Belnap-style. The proof of completeness is relegated to Section \ref{sec:completeness}.


%% file: syntax_semantics_flatness.tex
In the present section, we briefly recall basic definitions and facts about inquisitive logic, and refer the reader to \cite{InquiLog,Ciardelli_thesis} for an expanded treatment.

Although the support-based semantics (or team semantics) is originally developed for the extension of classical propositional logic with questions, for the sake of a better compatibility with the exposition in the next sections, we will first define support-based semantics (or team semantics) for classical propositional logic. Let us fix a set $\mathsf{Prop}$ of proposition variables, and denote its elements by  $p, q,\dots$ 
Well-formed
formulas of \emph{classical propositional logic} (\CPC), also called \emph{classical formulas}, are given by the following grammar:
\[
    \chi::= \,p\mid 0\mid\chi\iand\chi\mid\chi\ira\chi.
\]
As usual, we write $\ineg\chi$ for $\chi \ira 0$.

  A \emph{possible world} (or a \emph{valuation}) is a map $v: \mathsf{Prop}\to 2$, where $2:=\{0,1\}$. An information \emph{state} (also called a \emph{team}) is a set of possible worlds. 

\begin{definition}
\label{TS_cpc}
The {\em support} relation of a classical formula $\chi$ on a state $S$, denoted $S\models\chi$, is defined recursively   as follows: 
\begin{center}
\begin{tabular}{l c l}
$S\models p$ & iff & $v(p)=1$ for all $v\in S$\\
$S\models 0$ & iff & $S=\varnothing$\\
$S\models\chi\iand\xi$ & iff & $S\models\chi$ and $S\models\xi$\\
$S\models \chi\ira\xi$ & iff & for all $S'\subseteq S$, if $S'\models\chi$, then $S'\models\xi$\\
\end{tabular}
\end{center}
\end{definition}
%
An easy inductive proof shows that classical formulas $\chi$ are \emph{flat} (also called  \emph{truth conditional}); 
that is, for every  state $S$,
\begin{description}
\item[(Flatness Property)] $S\models\chi~~\mbox{ iff }~~ \{v\}\models \chi$ for any  $v\in S~~ \mbox{ iff } ~~ v(\chi)=1$ for any  $v\in S$.
\end{description}
%

Well-formed
formulas $\phi$ of \emph{inquisitive logic} (\Inql)  are given by expanding the language of \CPC with the connective $\ior$. Equivalently, these formulas can be defined by the following recursion:
\[
    \phi::= \,\chi\mid\phi\iand\phi\mid\phi\ira\phi\mid \phi\ior\phi.
\]
This two-layered presentation is slightly different but equivalent to the usual one. The reason why we are presenting it this way will be clear at the end of the following section, when we introduce a translation of \Inql-formulas  into a multi-type language.

\begin{definition}\label{TS_PT}
The {\em support relation} of formulas $\phi$ of \Inql on a state $S$, denoted $S\models\phi$, is defined analogously to the support of  classical formulas relative to the fragment shared by the two languages, and moreover:   
\smallskip

\noindent

\begin{center}
\begin{tabular}{l c l}
 $S\models \phi\ior\psi$ & iff & $S\models \phi$ or $S\models\psi$.
\end{tabular}
\end{center}
\smallskip

\noindent We write $\phi\models\psi$ if, for any  state $S$, if $S\models\phi$ then $S\models\psi$. If both $\phi\models\psi$ and $\psi\models\phi$, then we write $\phi\equiv\psi$. An \Inql-formula $\phi$  is \emph{valid}, denoted  $\models\phi$, if $S\models\phi$  for any  state $S$. The \emph{logic} \Inql is the set of all valid \Inql-formulas.
\end{definition}


An easy inductive proof shows that \Inql-formulas have 
the downward closure property and the empty team property:
\begin{description}
\item[(Downward Closure Property)] If $S\models\phi$ and $S'\subseteq S$, then $S'\models\phi$.
\item[(Empty Team Property)] $\varnothing\models\phi$.
\end{description}

\CPC extended with the dependence atoms $\dep(p_1,\dots,p_n,q)$ is called propositional dependence logic (\PD), which is an important variant of \Inql. \PD adopts also the state semantics (or the team semantics). It is proved in \cite{VY_PD} that \PD has the same expressive power as \Inql.  
In particular, a constancy dependence atom $\dep(p)$ is semantically equivalent to the formula $p\ior\ineg p$, which expresses the \emph{polar question} `whether $p$?' (denoted $?p$), and a dependence atom $\dep(p_1,\dots,p_n,q)$ with multiple arguments is semantically equivalent to the entailment $?p_1\wedge \dots\wedge ?p_n\ira ?q$ of polar questions. For more details on this connection, we refer the reader to \cite{Ciardelli2015}.

Flat formulas will play an important role in this paper. Below we list some of their properties.  

\begin{lemma}[see \cite{ivano_msc}]
For all \Inql-formulas $\phi$ and $\psi$,
\begin{itemize}
\item If $\psi$ is flat, then $\phi\ira\psi$ is flat. In particular, $\ineg\phi$ is always flat.
\item The following are equivalent:
\begin{enumerate}
\item $\phi$ is flat.
\item $\phi\equiv\phi^{\mathsf{f}}$, where $\phi^{\mathsf{f}}$ is the classical formula obtained from $\phi$ by replacing every occurrence of $\phi_1\ior \phi_2$ in $\phi$ by $\ineg \phi_1 \ira \phi_2$.
\item $\phi\equiv \ineg\ineg\phi$.
\end{enumerate}
\end{itemize}
\end{lemma}

%
Below we list some meta-logical properties of \Inql; for the proof, see \cite{InquiLog}. For any set  $\Gamma\cup\{\phi,\psi\}$ of \Inql-formulas:
\begin{description}
\item[(Deduction Theorem)] \(\Gamma,\phi\models\psi\text{ if and only if }\Gamma\models\phi\to\psi.\)
\item[(Disjunction Property)] If $\models\phi\ior\psi$, then either $\models\phi$ or $\models\psi$.
\item[(Compactness)] If $\Gamma\models\phi$, then there exists a finite subset $\Delta$ of $\Gamma$ such that $\Delta\models\phi$.
\end{description}
%

\begin{theorem}[see \cite{InquiLog,Ciardelli_thesis}]
\label{thm:fan:soundcomplete}
The following Hilbert-style system of \Inql is sound and complete.

\begin{description}
\item[Axioms:] \
\begin{enumerate}
\item all substitution instances of \IPC axioms
\item $(\chi \ira (\phi \ior \psi)) \ira (\chi \ira \phi) \ior (\chi \ira \psi)$ whenever $\chi$ is a classical formula
\item $\ineg\ineg \chi\ira \chi$  whenever $\chi$ is a classical formula
\end{enumerate}


\item[Rule:] \
\begin{description}
\item[ Modus Ponens:] \AxiomC{$\phi\to \psi$} \AxiomC{$\psi$}\BinaryInfC{$\psi$} \DisplayProof ~(\MP)
\end{description}
\end{description}
\end{theorem}

Clearly, the syntax of \Inql is the same as that of intuitionistic propositional logic (\IPC), but the connections between inquisitive and intuitionistic logic are in fact much deeper. Indeed, it was proved in \cite{InquiLog} that for every intermediate logic $\mathsf{L}$, 
\footnote{ Recall that $\mathsf{L}$ is an intermediate logic if $\IPC\subseteq \mathsf{L}\subseteq \CPC$.} letting $\mathsf{L}^\neg=\{\phi\mid\phi^\neg\in\mathsf{L}\}$ be the \emph{negative variant} of $\mathsf{L}$, where $\phi^\neg$ is obtained from $\phi$ by replacing any occurrence of a propositional variable $p$ with $\ineg p$, then \Inql coincides with the negative variant of every intermediate logic that is between Maksimova's logic $\mathsf{ND}$ \cite{MaksimovaLog86} and Medvedev's logic $\mathsf{ML}$ \cite{MedvedevLog}, such as the Kreisel-Putnam logic $\mathsf{KP}$ \cite{KrsPutnamLog57}.


\begin{theorem}[see \cite{InquiLog}]
For any intermediate logic $\mathsf{L}$ such that $\mathsf{ND}\subseteq \mathsf{L}\subseteq \mathsf{ML}$, we have $\mathsf{L}^\neg=\Inql$. In particular, $\Inql=\mathsf{KP}^\neg=\mathsf{ND}^\neg=\mathsf{ML}^\neg$.
\end{theorem}



%% file: semanticanalysis.tex
In the present section, building on \cite{AbramskyVaananen08, Roelofsen_algebraic_13}, and using standard facts pertaining to discrete Stone and Birkhoff dualities, we give an alternative algebraic presentation of the team semantics. This presentation shows how two natural types emerge from the team semantics, together with natural maps connecting them. These maps will  support the interpretation of additional {\em multi-type} connectives which will be used to define a new, multi-type language into which we will translate the original language and axioms of inquisitive logic. Finally, in Section \ref{sec:formal} we will introduce a structural  multi-type sequent calculus for the translated axiomatization.

\subsection{Order-theoretic analysis}
\label{ssec:semanticanalysis}


In what follows, we let $V$ abbreviate the initial set $\mathsf{Prop}$ of proposition variables; we let $2^V$ denote the set of Tarski assignments. Elements of $2^V$ are denoted by the variables $u$ and $v$, possibly sub- and super-scripted. Let $\mathbb{B}$ denote the (complete and atomic) Boolean algebra $(\mathcal{P}(2^V), \cap, \cup, (\cdot)^c,\varnothing, 2^V)$. Elements of $\mathbb{B}$ are information states (teams), and are denoted by the variables $S, T$ and $U$, possibly sub- and super-scripted. Consider the relational structure $\mathcal{F} = (\mathcal{P}(2^V), \subseteq)$
By discrete Birkhoff-type duality, a perfect Heyting algebra\footnote{A Heyting algebra is {\em perfect} if it is complete, completely distributive and completely join-generated by its completely join-prime elements. Equivalently, any perfect algebra can be characterized up to isomorphism as the complex algebra of some partially ordered set.} 
arises as the complex algebra of $\mathcal{F}$. Indeed,
let $\mathbb{A}: = (\mathcal{P}^{\downarrow}(\mathbb{B}), \cap, \cup, \Rightarrow, \varnothing, \mathcal{P}(2^V))$. Elements of $\mathbb{A}$ are downward closed collections of teams, and are denoted by the variables $\mathcal{X}, \mathcal{Y}$ and $\mathcal{Z}$, possibly sub- and super-scripted. The operation $\Rightarrow$ is defined as follows: for any $\mathcal{Y}$ and $\mathcal{Z}$,
\begin{align*}
\mathcal{Y}\Rightarrow\mathcal{Z} & := \{S\mid \mbox{ for all } S', \mbox{ if } S'\subseteq S \mbox{ and } S'\in \mathcal{Y}, \mbox{ then } S'\in \mathcal{Z}\}.
\end{align*}




Three natural maps can be defined between the perfect Boolean algebra $\mathbb{B}$ and the perfect HAO $\mathbb{A}$. Indeed, any team $S$ can be associated with the downward-closed collection of teams ${\downarrow}S: = \{S'\mid S'\subseteq S\}$. Conversely, any (downward-closed) collection of teams $\mathcal{X}$ can be associated with the team $\hb\mathcal{X} := \bigcup\mathcal{X} = \{v\mid v\in S$  for some $S\in \mathcal{X}\}.$ Thirdly, for any team $S$, the collection of teams  $\hb^\ast: = \{\{v\}\mid v\in X\}\cup\{\varnothing\}$ is downward closed. These assignments respectively define the following maps:
\[\bh: \mathbb{B}\to \mathbb{A}\quad\quad \hb: \mathbb{A}\to \mathbb{B}\quad \quad \hb^{\ast}: \mathbb{B}\to \mathbb{A}.\]

The maps $\hb^{\ast}$, $\bh$ and $\hb$ turn out to be adjoints to one another as follows:\footnote{In order-theoretic notation we write  $\hb^{\ast}\dashv \hb\dashv \bh$).}

\begin{lemma}
\label{lemma: bh left adjoint of hb}
For all $S\in \mathbb{B}$ and $\mathcal{X}\in \mathbb{A}$,
\begin{align}
\hb \mathcal{X}\subseteq S \quad \mbox{ iff }\quad \mathcal{X}\subseteq \bh S && \text{and} &&
\hb^{\ast}S\subseteq \mathcal{X} \quad \mbox{ iff }\quad S\subseteq \hb\mathcal{X}.
\end{align}
\end{lemma}
By general order-theoretic facts, from these adjunctions it follows that $\bh$, $\hb$ and $\hb^*$ are all order-preserving (monotone), and moreover, $\bh$ preserves all meets of $\mathbb{B}$ (including the empty one, i.e.\ $\bh 1^{\mathbb{B}} = \top^{\mathbb{A}}$), that is, $\bh$ commutes with arbitrary intersections, $\hb$ preserves all joins and all meets of $\mathbb{A}$, that is, $\hb$ commutes with arbitrary unions and intersections, and $\hb^*$ preserves all joins  of $\mathbb{B}$, that is, $\hb$ commutes with arbitrary unions.
Notice also that  for all $\mathcal{X}\in \mathbb{A}$ and $S, T\in \mathbb{B}$,
\begin{equation}
\label{eq:soundness of rules}
\mathcal{X}\subseteq \bh \hb(\mathcal{X})\quad \mbox{ and }\quad S\subseteq T \ \mbox{ implies } \hb^*(S)\subseteq \bh T.
\end{equation}

The following lemma  will be needed to prove the soundness of the rule KP of the calculus introduced in section \ref{sec:formal}.

\begin{lemma}
\label{lemma: soundness of rules}
For all $X$, $\mathcal{Y},\mathcal{Z}$,
\begin{center} \label{lem:sem:item2}
$\bh X\Rightarrow (\mathcal{Y}\cup \mathcal{Z})\subseteq (\bh X\Rightarrow \mathcal{Y})\cup (\bh X\Rightarrow \mathcal{Z})$;
\end{center}
\end{lemma}
\begin{proof}
%
Assume that $W\in \bh X\Rightarrow (\mathcal{Y}\cup \mathcal{Z})$ and $W\notin \bh X\Rightarrow \mathcal{Z}$. Then $W'\subseteq X$ and $W'\notin \mathcal{Z}$ for some $W'\subseteq W$. Hence $W\notin \mathcal{Z}$. To show that $W\in \bh X\Rightarrow \mathcal{Y}$, let $Z\subseteq W\cap X$. Then by assumption, either  $Z\in \mathcal{Y}$ or $Z\in \mathcal{Z}$. However, $W\notin \mathcal{Z}$ implies that $Z\notin \mathcal{Z}$, and hence $Z\in \mathcal{Y}$, as required.
%
\end{proof}

The following lemma collects  relevant  properties of $\bh$:
\begin{lemma}
\label{lemma:properties of bh}
For all $X,Y\in \mathbb{B}$,
\begin{itemize}
\item[(a)] $\bh \bot_\mathbb{B}=\{\varnothing\}$  
and $\bh \top^\mathbb{B}=\top^\mathbb{A}$;
\item[(b)] $\bh(\bigcap_{i\in I}X_i)=\bigcap_{i\in I}\bh X_i$;
\item[(c)] $\bh(X^c\cup Y)=(\bh X)\Rightarrow (\bh Y)$.
\end{itemize}
\end{lemma}

 \begin{proof}
(a) Immediate.\\

\noindent 
\begin{tabular}{r c l}
(b) $\bh(\bigcap_{i\in I} X_i)$ &$ =$&$\{Z\mid Z\subseteq \bigcap_{i\in I} X_i\}$\\
&$ =$&$\{Z\mid  Z\subseteq X_i \mbox{ for all } i\in I\}$\\
&$ =$&$\{Z\mid  Z\in \bh X_i \text{ for all } i\in I\}$\\
&$ =$&$\bigcap_{i\in I} (\bh X_i).$\\
\end{tabular}


\noindent 
\begin{tabular}{r c l}
(c) $(\bh X)\Rightarrow (\bh Y)$ &$ =$&$\{Z\mid \mbox{for any } W, \mbox{ if } W\subseteq Z \mbox{ and } W\subseteq X \mbox{ then } W\subseteq Y\}$\\
&$ =$&$\{Z\mid  \mbox{ if }  Z\subseteq X \mbox{ then } Z\subseteq Y\}$\\
&$ =$&$\{Z\mid   Z\subseteq X^c\cup Y\}$\\
&$ =$&$\bh ( X^c\cup Y).$
\end{tabular}
\end{proof}

\subsection{Multi-type inquisitive logic}
\label{ssec:multi-type Inql}
The existence of the maps $\bh$, $\hb$ and $\hb^*$ motivates the introduction of the following language, the formulas of which are given in two types, $\mathsf{Flat}$ and $\mathsf{General}$, defined by the following simultaneous recursion:
\begin{center}
$\mathsf{Flat}\ni \alpha ::= \,p \mid 0 \mid \alpha \cand \alpha \mid \alpha \cra \alpha \quad \quad \mathsf{General}\ni  A ::= \,\bh \alpha \mid A \iand A \mid A \ior A \mid A \ira A$
\end{center}

Let $\cneg\alpha$ and $\alpha\cor\beta$ abbreviate $\alpha\cra 0$ and  $\cneg\alpha\cra\beta$ respectively.
Notice that a canonical assignment exists $\hat{\cdot}:\mathsf{Prop}\rightarrow \mathbb{B}$, defined by $p\mapsto \hat{p} := \{v\mid v(p) = 1\}$. This assignment can be extended to $\mathsf{Flat}$-formulas as usual via the homomorphic extension $\val{\cdot}_\mathbb{B}: \mathsf{Flat}\to \mathbb{B}$.
The homomorphic extension $\val{\cdot}_\mathbb{B}: \mathsf{Flat}\to \mathbb{B}$ can be composed with $\bh: \mathbb{B}\to \mathbb{A}$ so as to yield a second homomorphic extension $\val{\cdot}_\mathbb{A}: \mathsf{General}\to \mathbb{A}$.
The maps $\val{\cdot}_\mathbb{B}$ and
$\val{\cdot}_\mathbb{A}$ are defined as below:

\begin{center}
\begin{tabular}{r c l c r c l}
$\val{p}_\mathbb{B}$ &$ = $& $\hat{p}$
& \quad\quad\quad &
$\val{\bh\alpha}_\mathbb{A}$ &$ = $& $\bh\val{\alpha}_\mathbb{B}$\\
$\val{0}_\mathbb{B}$ &$ = $& $\varnothing$
&&
$\val{A\ior B}_\mathbb{A}$ &$ = $& $\val{A}_\mathbb{A}\cup\val{B}_\mathbb{A}$\\
$\val{\alpha\cand \beta}_\mathbb{B}$ &$ = $& $\val{\alpha}_\mathbb{B}\cap\val{\beta}_\mathbb{B}$
&&
$\val{A\iand B}_\mathbb{A}$ &$ = $& $\val{A}_\mathbb{A}\cap\val{B}_\mathbb{A}$\\
$\val{\alpha\cra \beta}_\mathbb{B}$ &$ = $& $(\val{\alpha}_\mathbb{B})^c\cup
\val{\beta}_\mathbb{B}$
&&
$\val{A\ira B}_\mathbb{A}$ &$ = $& $\val{A}_\mathbb{A}\Rightarrow\val{B}_\mathbb{A}$.\\
$\val{\alpha\cor \beta}_\mathbb{B}$ &$ = $& $\val{\alpha}_\mathbb{B}\cup\val{\beta}_\mathbb{B}$
&&
\\
\end{tabular}
\end{center}
The following lemma is an immediate consequence of the definitions of $\val{\cdot}_\mathbb{B}$ and  $\val{\cdot}_\mathbb{A}$, and of Lemma \ref{lemma:properties of bh}:
\begin{lemma}
For all $\mathsf{Flat}$-formulas $\alpha$ and $\beta$,
\begin{center}
\begin{tabular}{r c l c r c l}
$\val{\bh p}_\mathbb{A}$ &$ = $& $\bh\hat{p}$
& \quad\quad\quad\quad\quad &
$\val{\bh(\alpha\cand \beta)}_\mathbb{A}$ &$ = $& $\bh \val{\alpha}_\mathbb{B}\cap\bh
\val{\beta}_\mathbb{B}$\\
$\val{\bh 0}_\mathbb{A}$ &$ = $& $\{\varnothing\}$
&&
$\val{\bh(\alpha\cra \beta)}_\mathbb{A}$ &$ = $& $\bh \val{\alpha}_\mathbb{B}\Rightarrow\bh
\val{\beta}_\mathbb{B}$.\\
\end{tabular}
\end{center}
\end{lemma}

Let us define the multi-type counterpart of flat formulas of  inquisitive logic: 
\begin{definition}
A formula $A\in \mathsf{General}$ is {\em flat} if for every team $S$, \[S\models A \quad \mbox{ iff }\quad \{v\}\models A\ \mbox{ for every } v\in S.\]
\end{definition}  
\begin{lemma}
\label{lemma:semantic flatness}
The following are equivalent for any $A\in \mathsf{General}$:\\
1. $A$ is flat;\\
2. $\val{A}_{\mathbb{A}} = \bh \hb(\val{A}_{\mathbb{A}})$.
\end{lemma}
\begin{proof}
By definition, $A$ is flat iff  $\val{A}_{\mathbb{A}} = \{S\mid \hb^{\ast}(S)\subseteq \val{A}_{\mathbb{A}}\}$.
Moreover, the following chain of identities holds:
\begin{center}
\begin{tabular}{r c l l}
& & $\{X\mid \hb^{\ast}(X)\subseteq \val{A}_{\mathbb{A}}\}$\\
&$=$ & $\{X\mid X\subseteq \hb(\val{A}_{\mathbb{A}})\}$ & (Lemma \ref{lemma: bh left adjoint of hb})\\
&$=$ & $\bh \hb(\val{A}_{\mathbb{A}})$,\\
\end{tabular}
\end{center}
which completes the proof.
\end{proof}

We are now in a position to define the following translation of \Inql-formulas into formulas of the multi-type language introduced above:
\CPC-formulas $\chi$ and $\xi$ will be translated into $\mathsf{Flat}$-formulas via $\tau_c$, and \Inql-formulas $\phi$ and $\psi$ into $\mathsf{General}$-formulas via $\tau_i$ as follows:

\begin{center}
\begin{tabular}{r c l c r c l}
$\tau_c(p)$ &$=$& $p$ &$\quad\quad \quad$& $\tau_i(\chi)$ &$=$& $\bh\tau_c(\chi)$\\
$\tau_c(0)$ &$=$& $0$ && $\tau_i(\phi\ior \psi)$ &$=$& $\tau_i(\phi)\ior\tau(\psi)$\\
$\tau_c(\chi\iand \xi)$ &$=$& $\tau_c(\chi)\cand\tau(\xi)$ && $\tau_i(\phi\iand \psi)$ &$=$& $\tau_i(\phi)\iand\tau_i(\psi)$ \\
$\tau_c(\chi\ira \xi)$ &$=$& $\tau_c(\chi)\cra\tau(\xi)$ && $\tau_i(\phi\ira \psi)$ &$=$& $\tau_i(\phi)\ira\tau_i(\psi)$. \\
\end{tabular}
\end{center}

The translation above justifies the introduction of the following Hilbert-style presentation of the logic which is the natural multi-type counterpart of \Inql:
\begin{itemize}
\item Axioms 
\subitem (A1) \CPC axiom schemata for $\mathsf{Flat}$-formulas;
\subitem (A2) \IPC axiom schemata for $\mathsf{General}$-formulas;
\subitem (A3) $(\bh\alpha \ira (A \ior B)) \ira (\bh\alpha \ira A) \ior (\bh\alpha \ira B)$ 
\subitem (A4) $\ineg \ineg \bh \alpha \ira \bh \alpha$.
\end{itemize}
plus Modus Ponens rules for both $\mathsf{Flat}$-formulas and $\mathsf{General}$-formulas.

In the following section, we are going to introduce the calculus for this logic.

%% file: DisplayCalculusPorPropositionalDependenceLogic.tex
In the present section, we introduce the structural calculus for the multi-type inquisitive logic introduced at the end of Section \ref{ssec:multi-type Inql}.

\begin{itemize}
\item Structural and operational languages of type $\mathsf{Flat}$ and $\mathsf{General}$:



\begin{center}
\begin{tabular}{lcl}
$\mathsf{Flat}$ && $\mathsf{General}$\\
&&\\
$\Gamma ::= \,\Phi \mid \Gamma \,, \Gamma \mid \Gamma \CRA \Gamma\mid \HB X$ & & $X ::= \,\BH \Gamma\mid \HB^*\Gamma \mid X \,; X \mid X > X$ \\
 & & \\
$\alpha ::= \,p \mid 0 \mid \alpha \cand \alpha \mid \alpha \cra \alpha$ & & $A ::= \,\bh \alpha \mid A \iand A \mid A \ior A \mid A \ira A$ \\
\end{tabular}
\end{center}



\item Interpretation of structural $\mathsf{Flat}$ connectives as their operational (i.e.\ logical) counterparts:\footnote{\label{footnote:structural interpreted operational} We follow the notational conventions introduced in \cite{LORI}, according to which each structural connective in the upper row of the synoptic tables is interpreted as the logical connective(s) in the two slots below it in the lower row. Specifically, each of its occurrences in antecedent (resp.\ succedent) position is interpreted as the logical connective in the left-hand  (resp.\ right-hand) slot. Hence, for instance, the structural symbol $\CRA$ is interpreted as classical implication $\cra$ when occurring in succedent position and as classical disimplication $\cdra$  (i.e.\ $\alpha \cdra \beta : = \alpha \cand \cneg\beta$) when occurring in antecedent position.}

\begin{center}
\begin{tabular}{|r|c|c|c|c|c|c|}
\hline
\scriptsize{Structural symbols}    & \mc{2}{c|}{$\Phi$}     & \mc{2}{c|}{$,$}     & \mc{2}{c|}{$\sqsupset$} \\
\hline
\scriptsize{Operational symbols} & $(1)$ & $\phantom{(}0\phantom{)}$ & $\cand$ & $(\cor)$ & $(\cdra)$ & $\phantom{(}\cra\phantom{)}$ \\
\hline
\end{tabular}
\end{center}

\item Interpretation of structural $\mathsf{General}$ connectives as their operational    counterparts: 

\begin{center}
\begin{tabular}{|r|c|c|c|c|}
\hline
\scriptsize{Structural symbols}     & \mc{2}{c|}{$;$}   & \mc{2}{c|}{$>$}   \\
\hline
\scriptsize{Operational symbols} & $\iand$ & $\ior$ & $(\idra)$ & $\ira$ \\
\hline
\end{tabular}
\end{center}


\item Interpretation of multi-type connectives


\begin{center}
\begin{tabular}{|r|c|c|c|c|c|c|}
\hline
\scriptsize{Structural symbols}    & \mc{2}{c|}{$\HB^\ast$}                          & \mc{2}{c|}{$\HB$} & \multicolumn{2}{c|}{$\BH$} \\
\hline
\scriptsize{Operational symbols} & $(\hb^\ast)$ & $\phantom{(\hb^\ast)}$ & $(\hb)$ & $(\hb)$ & $\bh$ & $\bh$                     \\
\hline
\end{tabular}
\end{center}


\item Structural rules common to both types

\begin{center}
\begin{tabular}{cc}
\!\!\!\!\!\!
{\footnotesize{
\begin{tabular}{rl}
\mc{2}{c}{\ \ \ \ \ \ \ \
\AX$\Gamma \fCenter \alpha$
\AX$(\Sigma \fCenter \Delta)[\alpha]^{pre}$
\RightLabel{$Cut$}
\BI$(\Sigma \fCenter \Delta)[\Gamma \slash \alpha]^{pre}$
\DisplayProof
 }
 \\

 & \\

\AX$\Gamma \fCenter \Delta$
\doubleLine
\LeftLabel{$\Phi$}
\UI$\Phi \,, \Gamma \fCenter \Delta$
\DisplayProof
 &
\AX$\Gamma \fCenter \Delta$
\doubleLine
\RightLabel{$\Phi$}
\UI$\Gamma \fCenter \Phi \,, \Delta$
\DisplayProof
 \\

 & \\

\AX$\Gamma \fCenter \Delta$
\LeftLabel{$W$}
\UI$\Gamma \,, \Sigma \fCenter \Delta$
\DisplayProof &
\AX$\Gamma \fCenter \Delta$
\RightLabel{$W$}
\UI$\Gamma \fCenter \Delta \,, Z$
\DisplayProof \\

 & \\

\AX$\Gamma \,, \Gamma \fCenter \Delta$
\LeftLabel{$C$}
\UI$\Gamma \fCenter \Delta $
\DisplayProof &
\AX$\Gamma \fCenter \Delta \,, \Delta$
\RightLabel{$C$}
\UI$\Gamma \fCenter \Delta$
\DisplayProof \\

 & \\

\AX$\Gamma \,, \Delta \fCenter \Sigma$
\LeftLabel{$E$}
\UI$\Delta \,, \Gamma \fCenter \Sigma$
\DisplayProof &
\AX$\Gamma \fCenter \Delta \,, \Sigma$
\RightLabel{$E$}
\UI$\Gamma \fCenter \Sigma \,, \Delta$
\DisplayProof \\

 & \\

\AX$\Gamma \,, (\Delta \,, \Sigma) \fCenter \Pi$
\LeftLabel{$A$}
\UI$(\Gamma \,, \Delta) \,, \Sigma \fCenter \Pi$
\DisplayProof &
\AX$\Gamma \fCenter (\Delta \,, \Sigma) \,, \Pi$
\RightLabel{$A$}
\UI$\Gamma \fCenter \Delta \,, (\Sigma \,, \Pi)$
\DisplayProof \\

 & \\

\AX$(\Gamma \CRA \Delta) \,, \Sigma \fCenter \Pi$
\LeftLabel{G}
\UI$\Gamma \CRA (\Delta \,, \Sigma) \fCenter \Pi$
\DisplayProof
 &
\AX$\Pi \fCenter (\Gamma \CRA \Delta) \,, \Sigma$
\RightLabel{G}
\UI$\Pi \fCenter \Gamma \CRA (\Delta \,, \Sigma)$
\DisplayProof
 \\
\end{tabular}
}}

 &

{\footnotesize{
\begin{tabular}{rl}
\mc{2}{c}{\ \ \ \ \ \ \ \ \
\AX$X \fCenter A$
\AX$A \fCenter Y$
\RightLabel{$Cut$}
\BI$X \fCenter Y$
\DisplayProof
 }
 \\

 & \\

\AX$X \fCenter Y$
\doubleLine
\LeftLabel{$\BH \Phi$}
\UI$\BH \Phi \,; X \fCenter Y$
\DisplayProof
 &
\AX$X \fCenter Y$
\doubleLine
\RightLabel{$\BH \Phi$}
\UI$X \fCenter \BH \Phi \,; Y$
\DisplayProof
 \\

 & \\

\AX$X \fCenter Y$
\LeftLabel{$W$}
\UI$X \,; Z \fCenter Y$
\DisplayProof &
\AX$X \fCenter Y$
\RightLabel{$W$}
\UI$X \fCenter Y \,; Z$
\DisplayProof \\

 & \\

\AX$X \,; X \fCenter Y$
\LeftLabel{$C$}
\UI$X \fCenter Y $
\DisplayProof &
\AX$X \fCenter Y \,; Y$
\RightLabel{$C$}
\UI$X \fCenter Y$
\DisplayProof \\

 & \\

\AX$X \,; Y \fCenter Z$
\LeftLabel{$E$}
\UI$Y \,; X \fCenter Z$
\DisplayProof &
\AX$X \fCenter Y \,; Z$
\RightLabel{$E$}
\UI$X \fCenter Z \,; Y$
\DisplayProof \\

 & \\

\AX$X \,; (Y \,; Z) \fCenter W$
\LeftLabel{$A$}
\UI$(X \,; Y) \,; Z \fCenter W$
\DisplayProof &
\AX$X \fCenter (Y \,; Z) \,; W$
\RightLabel{$A$}
\UI$X \fCenter Y \,; (Z \,; W)$
\DisplayProof \\

 & \\

\AX$(X > Y) \,; Z \fCenter W$
\LeftLabel{G}
\UI$X > (Y \,; Z) \fCenter W$
\DisplayProof
 &
\AX$W \fCenter (X > Y) \,; Z$
\RightLabel{G}
\UI$W \fCenter X > (Y \,; Z)$
\DisplayProof
 \\
\end{tabular}
}}
 \\
\end{tabular}
\end{center}

\item Structural rules specific to the $\mathsf{Flat}$ type

\begin{center}
{\footnotesize
\begin{tabular}{cc}
\AXC{\phantom{$p \fCenter p$}}
\LeftLabel{$Id$}
\UI$p \fCenter p$
\DisplayProof
 &
\AX$\Pi \fCenter \Gamma \CRA (\Delta \,, \Sigma)$
\RightLabel{CG}
\UI$\Pi \fCenter (\Gamma \CRA \Delta) \,, \Sigma$
\DisplayProof
 \\
\end{tabular}
 }
\end{center}

\item Structural rules governing the interaction between the two types:
\end{itemize}


\begin{center}
{\footnotesize
\begin{tabular}{c}
\AX$\Gamma \fCenter \Delta$
\RightLabel{bal}
\UI$\HB^\ast \Gamma \fCenter \BH \Delta$
\DisplayProof

\qquad

\AX$\Gamma \fCenter \Delta$
\RightLabel{d mon}
\UI$\BH \Gamma \fCenter \BH \Delta$
\DisplayProof

\qquad

\AX$X \fCenter Y$
\RightLabel{f mon}
\UI$\HB X \fCenter \HB Y$
\DisplayProof

 \\
 \\

\AX$\HB^\ast \Gamma \fCenter \Delta$
\RightLabel{f adj}
\doubleLine
\UI$\Gamma \fCenter \HB \Delta$
\DisplayProof

\qquad

\AX$\HB X \fCenter \Gamma$
\RightLabel{d adj}
\doubleLine
\UI$X \fCenter \BH \Gamma$
\DisplayProof

\qquad

\AX$\BH \HB X \fCenter Y$
\RightLabel{d-f elim}
\UI$X \fCenter Y$
\DisplayProof

 \\

 \\

\AX$X \fCenter \BH (\Gamma \CRA \Delta)$
\RightLabel{d dis}
\doubleLine
\UI$X \fCenter \BH \Gamma > \BH \Delta$
\DisplayProof

\qquad

\AX$\HB X \,, \HB Y \fCenter Z$
\RightLabel{f dis}
\doubleLine
\UI$\HB (X \,; Y) \fCenter Z$
\DisplayProof

 \\
 \\

\AX$X \fCenter \BH \Gamma > (Y \,; Z)$
\AX$X \fCenter \BH \Gamma > (Y \,; Z)$
\RightLabel{KP}
\BI$X \fCenter (\BH \Gamma > Y) \,; (\BH \Gamma > Z)$
\DisplayProof
 \\

\end{tabular}
 }
\end{center}

\begin{itemize}
\item Introduction rules for pure-type logical connectives:

\begin{center}
{\footnotesize
\begin{tabular}{rlrl}
\AXC{\phantom{$\bot \fCenter \Phi$}}
\UI$0 \fCenter \Phi$
\DisplayProof
 &
\AX$\Gamma \fCenter \Phi$
\UI$\Gamma \fCenter 0$
\DisplayProof
 &
\AX$A \fCenter X$
\AX$B \fCenter Y$
\BI$A \ior B \fCenter X \,; Y$
\DisplayProof
 &
\AX$Z \fCenter A \,; B$
\UI$Z \fCenter A \ior B $
\DisplayProof
 \\

 & & & \\

\AX$\alpha \,, \beta \fCenter \Gamma$
\UI$\alpha \cand \beta \fCenter \Gamma$
\DisplayProof
 &
\AX$\Gamma \fCenter \alpha$
\AX$\Delta \fCenter \beta$
\BI$\Gamma \,, \Delta \fCenter \alpha \cand \beta$
\DisplayProof
 &
\AX$A \,; B \fCenter Z$
\UI$A \iand B \fCenter Z$
\DisplayProof
 &
\AX$X \fCenter A$
\AX$Y \fCenter B$
\BI$X \,; Y \fCenter A \iand B$
\DisplayProof
 \\

 & & & \\

\AX$\Gamma \fCenter \alpha$
\AX$\beta \fCenter \Delta$
\BI$\alpha \cra \beta \fCenter \Gamma \CRA \Delta$
\DisplayProof
 &
\AX$\Gamma \fCenter \alpha \CRA \beta$
\UI$\Gamma \fCenter \alpha \cra \beta $
\DisplayProof
 &
\AX$X \fCenter A$
\AX$B \fCenter Y$
\BI$A \ira B \fCenter X > Y$
\DisplayProof
 &
\AX$Z \fCenter A > B$
\UI$Z \fCenter A \ira B $
\DisplayProof
 \\
\end{tabular}
 }
\end{center}

\item Introduction rules for $\bh$:

\begin{center}
{\footnotesize
\begin{tabular}{c}
\AX$\BH \alpha \fCenter X$
\UI$\bh \alpha \fCenter X$
\DisplayProof
\quad
\AX$X \fCenter \BH \alpha$
\UI$X \fCenter \bh \alpha$
\DisplayProof
 \\
\end{tabular}
 }
\end{center}

\end{itemize}


%% file: Soundness.tex
In the present section, we discuss the soundness of the rules of the  calculus introduced in section \ref{sec:formal}, as well as its being able to capture flatness syntactically. The completeness of the calculus is discussed in section \ref{sec:completeness}

\subsection{Soundness}

As is typical of structural calculi, in order to prove the soundness of the rules,
structural sequents will be translated into operational sequents of the appropriate type, and operational sequents  will be interpreted
according to their type. Specifically,  each atomic proposition $p\in \mathsf{Prop}$ is assigned to the team $\val{p}: = \{v\in 2^V\mid v(p) = 1\}$.

In order to translate structures as operational terms, structural connectives need to be translated as logical connectives. To this effect,
 structural connectives are associated with one or more logical connectives,  and any given occurrence of a structural connective is translated as one or the other, according to its (antecedent or succedent) position, as indicated in the synoptic tables at the beginning of section \ref{sec:formal}. This procedure is completely standard, and is discussed in detail in \cite{GAV,Multitype,LORI}.

Sequents $A\vdash B$ (resp.\ $\alpha\vdash \beta$) will be interpreted as inequalities (actually inclusions) $\val{A}\leq \val{B}$ (resp.\ $\val{\alpha}\leq \val{\beta}$) in $\mathbb{A}$ (resp.\ $\mathbb{B}$); rules $(a_i\vdash b_i\mid i\in I)/c\vdash d$ will be interpreted as implications of the form ``if $\val{a_i}\subseteq \val{b_i}_Z$ for every $i\in I$, then $\val{c}\subseteq \val{d}$''.
Following this procedure, it is easy to see that:

\begin{itemize}
\item the soundness of  (d mon) and (f mon) follows from the monotonicity of the semantic operations $\bh$ and $\hb$ respectively (cf.\ discussion after Lemma \ref{lemma: bh left adjoint of hb});
    
\item  the soundness of  (d-f elim) and (bal) follows from  the observations in \eqref{eq:soundness of rules};
\item the soundness of (d adj) and (f adj) follows from Lemma \ref{lemma: bh left adjoint of hb};
\item the soundness of (f dis) follows from the fact that the semantic operation $\hb$ distributes over intersections;
\item the soundness of (d dis) follows from Lemma \ref{lemma:properties of bh} (c);
\item the soundness of  (KP)  follows from Lemma \ref{lemma: soundness of rules}.
\end{itemize}
The proof of the soundness of the remaining rules is well known and is omitted.
\subsection{Syntactic flatness captured by the calculus}
Lemma \ref{lemma:semantic flatness} provided a semantic identification of flat $\mathsf{General}$-formulas as those the extension of which is in the image of the semantic $\bh$. The following lemma provides a similar identification with syntactic means.
\begin{lemma}
If a formula is of the following shape $A ::= \bh \alpha \mid A \iand A \mid A \ira A$, then $A \dashv\vdash \bh \alpha$ for some $\alpha$.
\end{lemma}
\begin{proof}
Base case: $A = \bh \alpha$.

\begin{center}
{\footnotesize
\begin{tabular}{c}
\AX $\alpha \fCenter \alpha$
\UI $\BH \alpha \fCenter \BH \alpha$
\UI $\bh \alpha \fCenter \BH \alpha$
\UI $\bh \alpha \fCenter \bh \alpha$
\DisplayProof
\end{tabular}
 }
\end{center}

Inductive case 1: $A = B \iand C = \bh\beta \wedge \bh\gamma$ by induction hypothesis.

\begin{center}
{\footnotesize
\begin{tabular}{cc}
\AX$\alpha \fCenter \alpha$
\UI$\alpha \,, \beta \fCenter \alpha$
\UI$\alpha \cand \beta \fCenter \alpha$
\UI$\BH (\alpha \cand \beta) \fCenter \BH \alpha$
\UI$\BH (\alpha \cand \beta) \fCenter \bh \alpha$
\AX$\beta \fCenter \beta$
\UI$\alpha \cand \beta \fCenter \beta$
\UI$\alpha \,, \beta \fCenter \beta$
\UI$\BH (\alpha \cand \beta) \fCenter \BH \beta$
\UI$\BH (\alpha \cand \beta) \fCenter \bh \beta$
\BI$\BH (\alpha \cand \beta) \,; \BH (\alpha \cand \beta) \fCenter \bh \alpha \iand \bh \beta$
\UI$\BH (\alpha \cand \beta) \fCenter \bh \alpha \iand \bh \beta$
\UI$\bh (\alpha \cand \beta) \fCenter \bh \alpha \iand \bh \beta$
\DisplayProof
 &
\AX$\alpha \fCenter \alpha$
\UI$\BH \alpha \fCenter \BH \alpha$
\LeftLabel{d adj}
\UI$\HB \bh \alpha \fCenter \alpha$
\AX$\beta \fCenter \beta$
\UI$\BH \beta \fCenter \BH \beta$
\UI$\bh \beta \fCenter \BH \beta$
\LeftLabel{d adj}
\UI$\HB \bh \beta \fCenter \beta$
\BI$\HB \bh \alpha \,, \HB \bh \beta \fCenter \alpha \cand \beta$
\RightLabel{f dis}
\UI$\HB (\bh \alpha \,; \bh \beta) \fCenter \alpha \cand \beta$
\UI$\bh \alpha \,; \bh \beta \fCenter \BH \alpha \cand \beta$
\UI$\bh \alpha \,; \bh \beta \fCenter \bh (\alpha \cand \beta)$
\UI$\bh \alpha \iand \bh \beta \fCenter \bh (\alpha \cand \beta)$
\DisplayProof
 \\
\end{tabular}
 }
\end{center}

Inductive case 2: $A = B \ira C= \bh\beta \ira \bh\gamma$ by induction hypothesis.
\begin{center}
{\footnotesize
\begin{tabular}{cc}
\AX$\alpha \fCenter \alpha$
\UI$\BH \alpha \fCenter \BH \alpha$
\UI$\bh \alpha \fCenter \BH \alpha$
\LeftLabel{d adj}
\UI$\HB \bh \alpha \fCenter \alpha$
\AX$\beta \fCenter \beta$
\BI$\alpha \cra \beta \fCenter \HB \bh \alpha \CRA \beta$
\UI$\BH \alpha \cra \beta \fCenter \BH (\HB \bh \alpha \CRA \beta)$
\UI$\bh (\alpha \cra \beta) \fCenter \BH (\HB \bh \alpha \CRA \beta)$
\UI$\HB \bh (\alpha \cra \beta) \fCenter \HB \bh \alpha \CRA \beta$
\UI$\HB \bh \alpha \,, \HB \bh (\alpha \cra \beta) \fCenter \beta$
\LeftLabel{f dis}
\UI$\HB (\bh \alpha \,; \bh (\alpha \cra \beta)) \fCenter \beta$
\RightLabel{d adj}
\UI$\bh \alpha \,; \bh (\alpha \cra \beta) \fCenter \BH \beta$
\UI$\bh \alpha \,; \bh (\alpha \cra \beta) \fCenter \bh \beta$
\UI$\bh (\alpha \cra \beta) \fCenter \bh \alpha > \bh \beta$
\UI$\bh (\alpha \cra \beta) \fCenter \bh \alpha \ira \bh \beta$
\DisplayProof
 &
\AX$\alpha \fCenter \alpha$
\UI$\BH \alpha \fCenter \BH \alpha$
\UI$\BH \alpha \fCenter \bh \alpha$
\AX$\beta \fCenter \beta$
\UI$\BH \beta \fCenter \BH \beta$
\UI$\bh \beta \fCenter \BH \beta$
\BI$\bh \alpha \ira \bh \beta \fCenter \BH \alpha > \BH \beta$
\UI$\bh \alpha \ira \bh \beta \fCenter \BH (\alpha \CRA \beta)$
\LeftLabel{d adj}
\UI$\HB \bh \alpha \ira \bh \beta \fCenter \alpha \CRA \beta$
\UI$\HB \bh \alpha \ira \bh \beta \fCenter \alpha \cra \beta$
\RightLabel{d adj}
\UI$\bh \alpha \ira \bh \beta \fCenter \BH (\alpha \cra \beta)$
\UI$\bh \alpha \ira \bh \beta \fCenter \bh (\alpha \cra \beta)$
\DisplayProof

 \\
\end{tabular}
 }
\end{center}

\end{proof}

%% file: CutElimination.tex
In the present section, we prove that the calculus introduced in Section \ref{sec:formal} enjoys cut elimination and subformula property. Perhaps the most important feature of this calculus is that its cut elimination does not need to be proved brute-force, but can rather be inferred from a Belnap-style cut elimination meta-theorem, proved in \cite{TrendsXIII}, which holds for the so called {\em proper multi-type calculi}, the definition of which is reported below.
\subsection{Cut elimination meta-theorem for proper multi-type calculi}
\begin{theorem}
(cf.\ \cite[Theorem 4.1]{TrendsXIII}) Every proper multi-type calculus enjoys cut elimination and subformula property.
\end{theorem}

Proper multi-type calculi are those satisfying the following list of conditions:
\label{Trends:sec:quasi}


\paragraph*{C$_1$: Preservation of operational terms.\;} Each operational term occurring in a premise of an inference rule {\em inf} is a subterm of some operational term in the conclusion of {\em inf}.

\paragraph*{C$_2$: Shape-alikeness of parameters.\;} Congruent parameters (i.e.\ non-active terms in the application of a rule) are occurrences of the same structure.

\paragraph*{C$'_2$: Type-alikeness of parameters.\;}  Congruent parameters have exactly the same type. This condition bans the possibility that a parameter changes type along its history.

\paragraph*{C$_3$: Non-proliferation of parameters.\;} Each parameter in an inference rule {\em inf} is congruent to at most one constituent in the conclusion of {\em inf}.

\paragraph*{C$_4$: Position-alikeness of parameters.\;} Congruent parameters are either all precedent or all succedent parts of their respective sequents. In the case of calculi enjoying the display property, precedent and succedent parts are defined in the usual way (see \cite{Belnap}). Otherwise, these notions can still be defined by induction on the shape of the structures, by relying on the polarity of each coordinate of the structural connectives.

\paragraph*{C$'_5$: Quasi-display of principal constituents.\;} If an operational term $a$ is principal in the conclusion sequent $s$ of a derivation $\pi$, then $a$ is in display, unless $\pi$ consists only of its conclusion sequent $s$ (i.e.\ $s$ is an axiom).

\paragraph*{C$''_5$: Display-invariance of axioms.} If $a$ is principal in an axiom $s$, then $a$ can be isolated by applying Display Postulates and the new sequent is still an axiom.

\paragraph*{C$'_6$: Closure under substitution for succedent parts within each type.\;} Each rule is closed under simultaneous substitution of arbitrary structures for congruent operational terms occurring in succedent position, {\em within each type}.

\paragraph*{C$'_7$: Closure under substitution for precedent parts within each type.\;} Each rule is closed under simultaneous substitution of arbitrary structures for congruent operational terms occurring in precedent position, {\em within each type}.

\paragraph*{C$'_8$: Eliminability of matching principal constituents.\;}
This condition requests a standard Gentzen-style checking, which is now limited to the case in which  both cut formulas are {\em principal}, i.e.~each of them has been introduced with the last rule application of each corresponding subdeduction. In this case, analogously to the proof Gentzen-style, condition C$'_8$ requires being able to transform the given deduction into a deduction with the same conclusion in which either the cut is eliminated altogether, or is transformed in one or more applications of the cut rule, involving proper subterms of the original operational cut-term. In addition to this, specific to the multi-type setting is the requirement that the new application(s) of the cut rule be also {\em type-uniform} (cf.\ condition C$'_{10}$ below).

\paragraph*{C$'''_8$: Closure of axioms under surgical cut.} If $(x\vdash y)([a]^{pre}, [a]^{suc})$,  $a\vdash z [a]^{suc}$ and $v[a]^{pre}\vdash a$ are axioms, then $(x\vdash y)([a]^{pre}, [z/a]^{suc})$ and $(x\vdash y)([v/a]^{pre}, [a]^{suc})$  are again axioms.

\paragraph*{C$_9$: Type-uniformity of derivable sequents.} Each derivable sequent is type-uniform.\footnote{A sequent $x \vdash y$ is type-uniform if $x$ and $y$ are of the same type.}

\paragraph*{C$'_{10}$: Preservation of type-uniformity of cut rules.} All cut rules preserve type-uniformity.

\subsection{Cut elimination for the structural calculus for multi-type inquisitive logic}
To show that the calculus defined in Section \ref{sec:formal} enjoys cut elimination and subformula property, it is enough to show that it is a proper multi-type calculus, i.e., that verifies every condition in the list above.
All conditions except C$'_8$ are readily satisfied by inspection on the rules of the calculus. In what follows we verify C$'_8$.

 Condition C$'_8$ requires to check 
the cut elimination when both cut formulas are principal. Since principal formulas are always introduced in display, it is enough to show that applications of standard (rather than surgical) cuts can be either eliminated or replaced with (possibly surgical) cuts on formulas of strictly lower complexity. 

\paragraph*{Constant}

\begin{center}
{\footnotesize
\begin{tabular}{ccc}
\bottomAlignProof
\AXC{\ \ \ $\vdots$ \raisebox{1mm}{$\pi_1$}}
\noLine
\UI$\Gamma \fCenter \Phi$
\UI$\Gamma \fCenter 0$
\AX$0 \fCenter \Phi$
\BI$\Gamma \fCenter \Phi$
\DisplayProof
 & 
$\rightsquigarrow$
 & 
\bottomAlignProof
\AXC{\ \ \ $\vdots$ \raisebox{1mm}{$\pi_1$}}
\noLine
\UI$\Gamma \fCenter \Phi$
\DisplayProof 
\end{tabular}
 }
\end{center}

\paragraph*{Propositional variable}
\begin{center}
{\footnotesize
\begin{tabular}{ccc}
\bottomAlignProof
\AX$p \fCenter p$
\AX$p \fCenter p$
\BI$p \fCenter p$
\DisplayProof
 & 
$\rightsquigarrow$
 &
\bottomAlignProof
\AX$p \fCenter p$
\DisplayProof 
\end{tabular}
 }
\end{center}

\paragraph*{Classical conjunction}
\begin{center}
{\footnotesize
\begin{tabular}{ccc}
\bottomAlignProof
\AXC{\ \ \ $\vdots$ \raisebox{1mm}{$\pi_1$}}
\noLine
\UI$\Gamma \fCenter \alpha$
\AXC{\ \ \ $\vdots$ \raisebox{1mm}{$\pi_2$}}
\noLine
\UI$\Delta \fCenter \beta$
\BI$\Gamma \,, \Delta \fCenter \alpha \cand \beta$
\AXC{\ \ \ $\vdots$ \raisebox{1mm}{$\pi_3$}}
\noLine
\UI$\alpha \,, \beta \fCenter \Lambda$
\UI$\alpha \cand \beta \fCenter \Lambda$
\BI$\Gamma \,, \Delta  \fCenter \Lambda$
\DisplayProof

 & 
$\rightsquigarrow$
 & 
\bottomAlignProof
\AXC{\ \ \ $\vdots$ \raisebox{1mm}{$\pi_1$}}
\noLine
\UI$\Gamma \fCenter \alpha$
\AXC{\ \ \ $\vdots$ \raisebox{1mm}{$\pi_2$}}
\noLine
\UI$\Delta \fCenter \beta$
\AXC{\ \ \ $\vdots$ \raisebox{1mm}{$\pi_3$}}
\noLine
\UI$\alpha \,, \beta \fCenter \Lambda$
\UI$\beta \fCenter \alpha > \Lambda$
\BI$\Delta \fCenter \alpha > \Lambda$
\UI$\alpha \,,  \Delta\fCenter \Lambda$
\UI$\Delta \,, \alpha \fCenter \Lambda$
\UI$ \alpha \fCenter \Delta > \Lambda$
\BI$\Gamma \fCenter \Delta > \Lambda$
\UI$\Delta \,,\Gamma \fCenter \Lambda$
\UI$\Gamma \,,\Delta \fCenter \Lambda$
\DisplayProof

\end{tabular}
}
\end{center}

The cases for  $\cra$, $\iand$, $\ior$, $\ira$ are standard and similar to the one above.

\paragraph*{Downarrow}

\begin{center}
{\footnotesize 
\begin{tabular}{ccc}
\bottomAlignProof
\AXC{\ \ \ $\vdots$ \raisebox{1mm}{$\pi_3$}}
\noLine
\UI$X \fCenter \BH \alpha$
\UI$X \fCenter \bh \alpha$
\AXC{\ \ \ $\vdots$ \raisebox{1mm}{$\pi_3$}}
\noLine
\UI$\BH \alpha \fCenter Y$
\UI$\bh \alpha \fCenter Y$
\BI$ X\fCenter Y$
\DisplayProof

 & 
$\rightsquigarrow$
 & 

\bottomAlignProof
\AXC{\ \ \ $\vdots$ \raisebox{1mm}{$\pi_3$}}
\noLine
\UI$X \fCenter \BH \alpha$
\UI$F X \fCenter \alpha$
\AXC{\ \ \ $\vdots$ \raisebox{1mm}{$\pi_3$}}
\noLine
\UI$\BH \alpha \fCenter Y$
\BI$ \BH F X\fCenter Y$
\UI$ X\fCenter Y$
\DisplayProof
\end{tabular}
 }
\end{center}

%% file: conclusions.tex
The calculus introduced in the present paper is not a standard display calculus. This is due to the fact that, according to the order-theoretic analysis we gave, the axiom (A3) is not analytic inductive in the sense of \cite{corr-and-displ}. Hence, it is not possible to give a proper {\em display} calculus to the axiomatization of the multi-type inquisitive logic introduced in Section \ref{ssec:multi-type Inql}. In order to encode the (A3) axiom with a structural rule, we made the non standard choice of allowing the structural counterpart of $\bh$ in antecedent position, notwithstanding the fact that it is {\em not} a left adjoint. As a consequence, the display property does not hold for the  calculus introduced in the present paper. However, a generalization of the Belnap-style cut elimination meta-theorem holds which applies to it.  

Further directions of research will address the problem of extending this calculus to propositional dependence logic.

%% file: Completeness-Proof.tex
\begin{center}
{\tiny 
\begin{tabular}{c}

\!\!\!\!\!\!\!\!\!\!\!\!\!\!\!\!\!\!\!\!\!\!\!\!\!\!\!\!\!\!\!\!\!\!\!\!\!\!\!\!\!\!

\AX$\alpha \fCenter \alpha$
\UI$\alpha \fCenter 0 \,, \alpha$
\UI$\alpha \,, \Phi \fCenter 0 \,, \alpha$
\UI$\Phi \fCenter \alpha \CRA (0 \,, \alpha)$
\RightLabel{CG}
\UI$\Phi \fCenter (\alpha \CRA 0) \,, \alpha$
\UI$\Phi \fCenter \alpha \,, (\alpha \CRA 0)$
\UI$\alpha \CRA \Phi \fCenter \alpha \CRA 0$
\UI$\BH (\alpha \CRA \Phi) \fCenter \BH (\alpha \CRA 0)$
\RightLabel{d dis}
\UI$\BH (\alpha \CRA \Phi) \fCenter \BH \alpha > \BH 0$
\UI$\BH \alpha \,; \BH (\alpha \CRA \Phi) \fCenter \BH 0$
\UI$\BH \alpha \,; \BH (\alpha \CRA \Phi) \fCenter \bh 0$
\UI$\BH (\alpha \CRA \Phi) \,; \BH \alpha \fCenter \bh 0$
\UI$\BH \alpha \fCenter \BH (\alpha \CRA \Phi) > \bh 0$
\UI$\bh \alpha \fCenter \BH (\alpha \CRA \Phi) > \bh 0$
\UI$\BH (\alpha \CRA \Phi) \,; \bh \alpha \fCenter \bh 0$
\UI$\bh \alpha \,; \BH (\alpha \CRA \Phi) \fCenter \bh 0$
\UI$\BH (\alpha \CRA \Phi) \fCenter \bh \alpha > \bh 0$
\UI$\BH (\alpha \CRA \Phi) \fCenter \bh \alpha \ira \bh 0$
\RightLabel{def}
\UI$\BH (\alpha \CRA \Phi) \fCenter \ineg \bh \alpha$

\AX$0 \fCenter \Phi$
\RightLabel{d mon}
\UI$\BH 0 \fCenter \BH \Phi$
\UI$\bh 0 \fCenter \BH \Phi$

\BI$\ineg \bh \alpha \ira \bh 0 \fCenter \BH (\alpha \CRA \Phi) > \BH \Phi$
\LeftLabel{def}
\UI$\ineg \ineg \bh \alpha \fCenter \BH (\alpha \CRA \Phi) > \BH \Phi$
\RightLabel{d dis}
\UI$\ineg \ineg \bh \alpha \fCenter \BH ((\alpha \CRA \Phi) \CRA \Phi)$
\RightLabel{d adj}
\UI$\HB \ineg \ineg \bh \alpha \fCenter (\alpha \CRA \Phi) \CRA \Phi$
\UI$(\alpha \CRA \Phi) \,, \HB \ineg \ineg \bh \alpha \fCenter \Phi$
\LeftLabel{G}
\UI$\alpha \CRA (\Phi \,, \HB \ineg \ineg \bh \alpha) \fCenter \Phi$
\UI$\Phi \,, \HB \ineg \ineg \bh \alpha \fCenter \alpha \,, \Phi$
\UI$\HB \ineg \ineg \bh \alpha \fCenter \alpha \,, \Phi$
\UI$\HB \ineg \ineg \bh \alpha \fCenter \alpha$
\RightLabel{d adj}
\UI$\ineg \ineg \bh \alpha \fCenter \BH \alpha$
\UI$\ineg \ineg \bh \alpha \fCenter \bh \alpha$
\DisplayProof


\!\!\!\!\!\!\!\!\!\!\!\!\!\!\!\!\!\!\!\!\!\!\!\!\!\!\!\!\!\!\!\!\!\!\!\!\!\!\!\!\!\!\!\!\!\!\!\!\!\!\!\!\!\!\!\!\!\!\!\!\!\!\!\!\!

\AX$\alpha \fCenter \alpha$
\RightLabel{d mon}
\UI$\BH \alpha \fCenter \BH \alpha$
\UI$\BH \alpha \fCenter \bh \alpha$
\UI$\bh \alpha \fCenter \bh \alpha$
\AX$B \fCenter B$
\AX$C \fCenter C$
\BI$B \ior C \fCenter B\,; C$
\BI$\bh \alpha \ira (B \ior C) \fCenter \BH \alpha > (B\,; C)$
\AX$\alpha \fCenter \alpha$
\RightLabel{d mon}
\UI$\BH \alpha \fCenter \BH \alpha$
\UI$\BH \alpha \fCenter \bh \alpha$
\AX$B \fCenter B$
\AX$C \fCenter C$
\BI$B \ior C \fCenter B\,; C$
\BI$\bh \alpha \ira (B \ior C) \fCenter \BH \alpha > (B\,; C)$
\RightLabel{KP}
\BI$\bh \alpha \ira (B \ior C) \fCenter (\BH \alpha > B)\,; (\BH \alpha > C)$
\UI$(\BH \alpha > B) > \bh \alpha \ira (B \ior C) \fCenter \BH \alpha > C$
\UI$\BH \alpha \,; ((\BH \alpha > B) > \bh \alpha \ira (B \ior C)) \fCenter C$
\UI$((\BH \alpha > B) > \bh \alpha \ira (B \ior C)) \,; \BH \alpha \fCenter C$
\UI$\BH \alpha \fCenter ((\BH \alpha > B) > \bh \alpha \ira (B \ior C)) > C$
\UI$\bh \alpha \fCenter ((\BH \alpha > B) > \bh \alpha \ira (B \ior C)) > C$
\UI$((\BH \alpha > B) > \bh \alpha \ira (B \ior C)) \,; \bh \alpha \fCenter C$
\UI$\bh \alpha \,; ((\BH \alpha > B) > \bh \alpha \ira (B \ior C)) \fCenter C$

\UI$(\BH \alpha > B) > \bh \alpha \ira (B \ior C) \fCenter \bh \alpha > C$
\UI$(\BH \alpha > B) > \bh \alpha \ira (B \ior C) \fCenter \bh \alpha \ira C$
\UI$\bh \alpha \ira (B \ior C) \fCenter (\BH \alpha > B) \,; \bh \alpha \ira C$
\UI$\bh \alpha \ira (B \ior C) \fCenter \bh \alpha \ira C \,; (\BH \alpha > B)$
\UI$\bh \alpha \ira C > \bh \alpha \ira (B \ior C) \fCenter \BH \alpha > B$
\UI$\BH \alpha \,; (\bh \alpha \ira C > \bh \alpha \ira (B \ior C)) \fCenter B$
\UI$(\bh \alpha \ira C > \bh \alpha \ira (B \ior C)) \,; \BH \alpha \fCenter B$
\UI$\BH \alpha \fCenter (\bh \alpha \ira C > \bh \alpha \ira (B \ior C)) > B$
\UI$\bh \alpha \fCenter (\bh \alpha \ira C > \bh \alpha \ira (B \ior C)) > B$
\UI$(\bh \alpha \ira C > \bh \alpha \ira (B \ior C)) \,; \bh \alpha \fCenter B$
\UI$\bh \alpha \,; (\bh \alpha \ira C > \bh \alpha \ira (B \ior C)) \fCenter B$
\UI$\bh \alpha \ira C > \bh \alpha \ira (B \ior C) \fCenter \bh \alpha > B$

\UI$\bh \alpha \ira C > \bh \alpha \ira (B \ior C) \fCenter \bh \alpha \ira B$
\UI$\bh \alpha \ira (B \ior C) \fCenter \bh \alpha \ira C \,; \bh \alpha \ira B$
\UI$\bh \alpha \ira (B \ior C) \fCenter \bh \alpha \ira B \,; \bh \alpha \ira C$
\UI$\bh \alpha \ira (B \ior C) \fCenter (\bh \alpha \ira B) \ior (\bh \alpha \ira C)$
\DisplayProof
\end{tabular}
 }
\end{center} 